\documentclass[11pt,oneside,reqno]{amsart}

 \usepackage[colorlinks=true,linkcolor=red,citecolor=blue]{hyperref}
 \usepackage{xcolor}
 \usepackage[final]{graphicx}
\usepackage{amsfonts}
\usepackage{pdfsync}
\usepackage{amsmath}
\usepackage{amssymb}
\usepackage{amsthm}
\usepackage{caption}
\usepackage{enumerate} % customize enumerate
\usepackage{dcolumn} % align numbers in tabular on decimal point
\newtheorem{theorem}{Theorem}[section]
\newtheorem{lemma}[theorem]{Lemma}
\newtheorem{corollary}[theorem]{Corollary}
\newtheorem{proposition}[theorem]{Proposition}

\theoremstyle{remark}
\newtheorem{remark}{Remark}

\newcommand{\R}{{\mathord{\mathbb R}}}

\newcommand{\N}{{\mathord{\mathbb N}}}

\DeclareMathOperator{\supp}{supp}
\newcommand\numberthis{\addtocounter{equation}{1}\tag{\theequation}}

\newcommand{\T}{\mathbb{T}}

\newcommand{\id}{\mathbb{I}}

\newcommand{\di}{\,\mathrm{d}}

\newcommand{\abs}[1]{\left|#1\right|}
\newcommand{\norm}[1]{\left\|#1\right\|}

\title{Mixing in an anharmonic potential well}

\date{\today}

\author[M. Moreno]{Matías Moreno}
\address{Departamento de Ingenier\'ia Matem\'atica and Centro
de Modelamiento Matem\'atico (CNRS IRL 2807), Universidad de Chile, Beauchef 851, Santiago, Chile}
\email{mmoreno@dim.uchile.cl}
\author[P. Rioseco]{Paola Rioseco}
\email{paola.rioseco@uchile.cl}
\author[H. Van Den Bosch]{Hanne Van Den Bosch}
\email{hvdbosch@dim.uchile.cl}

\begin{document}

\maketitle

\begin{abstract}
    We prove phase-space mixing for solutions to Liouville's equation for integrable systems. Under a natural non-harmonicity condition, we obtain weak convergence of the distribution function with rate $\langle \mathrm{time} \rangle^{-1}$. In one dimension, we also study the case where this condition fails at a certain energy, showing that mixing still holds but with a slower rate. When the condition holds and functions have higher regularity, the rate can be faster. 
\end{abstract}

\bigskip

\noindent{\bf Acknowledgments.} We thank Olivier Sarbach and Jean Bricmont for interesting discussions during the preparation of this manuscript.
All authors received support from the Center for Mathematical Modeling (Universidad de Chile \& CNRS IRL 2807) through ANID/Basal projects  \#FB210005 and  \#ACE210010.
 P.R acknowledges partial support from Junior research fellowship from Erwin Schrödinger International Institute for Mathematics and Physics, University of Vienna. H.VDB. acknowledges partial support from ANID/Fondecyt project \#118--0355 and project France-Chile MathAmSud EEQUADDII 20-MATH-04.

\section{Introduction}
% We study the Liouville equation (collisionless Boltzmann equation) for particles in a potential well in $\R^d$. Throughout the paper, we assume that $V:\R^d \mapsto \R$ is a smooth potential with a single, non-degenerate minimum at the origin. The Liouville equation describes the evolution of a density $F$ on the phase space $\R^d \times \R^d$,
% and reads
% \begin{equation} \label{eq:liouville}
% \partial_t F(t,x,p) = - p \cdot \nabla_x  F + \nabla V \cdot \nabla_p F,
% \end{equation}
% where we use the convention that $\nabla_p$ and $\nabla_x$ refer to the the vector of $d$ partial derivatives with respect to the momentum (resp position) variables.

% To simplify the discussion, we will assume throughout that $V$ is of class $C^2$, that $V(0)=0 $ is absolute minimum, that $\operatorname{Hess}V(0)>0$, that $x \cdot \nabla V(x) >0$ if $x \neq 0$ and that $V(x) \to + \infty$ as $|x| \to +\infty$.
% This guarantees that constant-energy surfaces of $H= p^2/2 + V(x)$ are bounded and star-shaped with respect to to the origin. 

% If the system is integrable (for instance, if $d=1$ or if $V$ has spherical symmetry in dimensions $2,3$), the phase space can be parametrized by action-angle coordinates $(q,k) \in \T^d \times \R_+^d$. In these coordinates, the particle density satisfies the transport equation

We study Liouville's equation in action-angle coordinates
\begin{equation} \label{eq:transport}
   \partial_t f(t,q,k) + \omega(k) \cdot \nabla_q f(t,q,k) = 0,
\end{equation}
where $q \in \T^d$, the $d$-dimensional (flat) torus, are the \emph{angles} and $(k_1, \cdots k_n)\in K$ the conserved quantities in a suitable open set $K$. The function $\omega: K \mapsto \R^d$ gives the frequencies associated to each angle. We think of Liouville's equation as describing the evolution of a large number of gas molecules or collisionless kinetic gas without interactions. 
Liouville's theorem guarantees that, if the motion of a single particle in this system is integrable, there exist coodinates $(q,k)$ that bring the equation in the form~\eqref{eq:transport}. 

In classical mechanics, this occurs for instance in a potential well in one space-dimension
or for spherically symmetric potentials in dimensions $2, 3$. 
In these cases, in the \emph{physical} coordinates, Liouville's equation reads
\begin{equation} \label{eq:liouville}
\partial_t F(t,x,p) = - p \cdot \nabla_x  F + \nabla _x V \cdot \nabla_p F,
\end{equation}
which can be transformed (for an open set $K\subset \R^d$ of values of the conserved energy and angular momentum) into the form \eqref{eq:transport}.

But also in a relativistic context, geodesic motion in the Kerr family of space-times is integrable and Liouville's equation (or the collisionless Boltzmann equation) can be written in the form\eqref{eq:transport}.

If the system is anharmonic, in the sense that points with nearby energies move at different angular speeds $\omega(k)$, regular initial distributions will eventually stretch out to thin filaments that cover the region of phase space allowed by the conservation laws, as illustrated in Figure~\ref{Fig:EvolutionPHO}.

% Here, $a(k)$ denotes the area function and $\nabla_q$ the vector of $d$ partial derivatives with respect to the $q$-variables.
% Throughout the paper, we will use $f$ to refer to the distribution function in action-angle variables, and $F$ for the distribution function depending on position and momentum in Euclidean space. 
% \textcolor{violet}{Agregar la definición de $a$. [Hanne:] Traté de generalizar al caso de mas dimensiones pero no estoy segura de haberlo hecho bien.} 
% The action variables are defined by 
% \begin{eqnarray}
% k_j(x,p) = \frac{a(H(x,p))}{2\pi} = \frac{1}{2\pi} \oint _{C(E)} p_j \di x_j
% \end{eqnarray}

% where $C(E)$ is the (hyper)-surface in phase space with energy $E=|p|^2/2 + V(x)$. 
% The angle coordinates are the conjugate variables given by
% \begin{equation}
% q(x,p) := \frac{2\pi}{T(H(x,p))} \oint _{\gamma_x} \frac{dx}{p},
% \end{equation}

% %\cite{Arnold-Book}

% We think of Liouville's equation as describing the evolution of a (thermodynamically) large number of gas molecules or dust particles without interactions. 
% If the system is anharmonic in the sense that points with nearby energies move at different angular speeds, regular initial distributions will eventually stretch out to tin filaments that cover all of phase space, as illustrated in Figure~\ref{Fig:EvolutionPHO}. 

\begin{figure}[ht]
\centerline{
\resizebox{5cm}{!}{\includegraphics{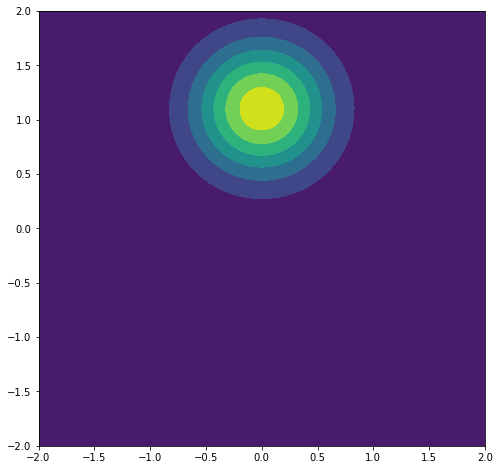}}
\resizebox{5cm}{!}{\includegraphics{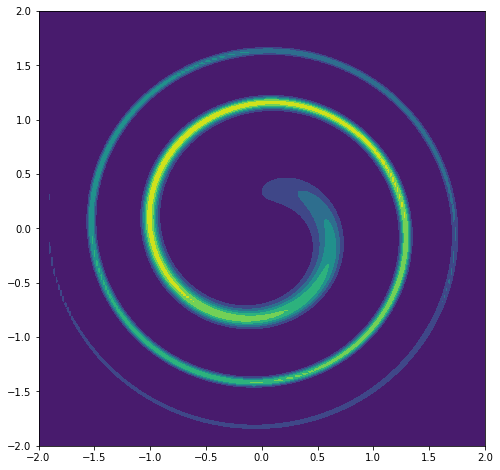}}
\resizebox{5cm}{!}{\includegraphics{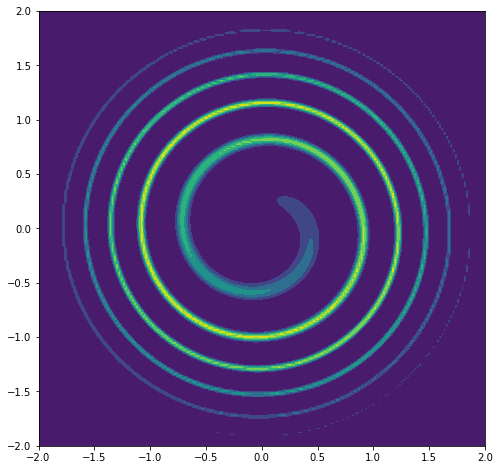}}
}
\caption{\label{Fig:EvolutionPHO} Snapshots at times $t= 0$, $40\pi$, $80 \pi $ of the evolution of a Gaussian initial condition in a perturbed harmonic oscillator with Hamiltonian $H= p^2/2 + x^2/2 + 0.3 \, x^4 $, aproximate at first order in perturbation theory.  
% distribution function corresponding to initial datum $f_0\left(x, p \right)= e^{-\frac{(x - \tilde{x}_0)}{\sigma}^2} e^{-\frac{(p - p_0)}{\sigma}^2}, \quad (x,p) \in \R^2 $ with values of $\sigma= 0.6$, $p_0=1.1 $ and $\tilde{x}_0 =0 $ and the value of perturbation $\epsilon = 0.3$ in the perturbed harmonic oscillator model $H(x,p)= p^2/2 + x^2/2 + \epsilon x^4 $ .
}
\end{figure}

This phenomenon is called phase-space mixing.
It leads to weak convergence in the sense that, for any measurement of a macroscopic quantity, encoded in a \emph{test function}\footnote{Throughout this paper, we will use the term ``test function" loosely to designate ``a function against which the distribution function is tested", these functions need not to belong to $C_c^\infty$. Mathematically, the roles of $\phi$ and $f$ are symmetric by the time-reversal invariance of the evolution, while conceptually $f$ is an empirical density associated to a large number of particles, and only its averages over regions of phase space (i.e., the support of $\phi$) have a physical significance.} 
$\phi(q,k)$, its value satisfies
$$
\lim_{t \to \infty}  \int_{K}\int_{\T^d}
f(t,q,k) \phi(q,k) \di q \di k =   \int_{K}
\bar f_0(k) \int_{\T^d} \phi(q,k) \di q \di k.
$$
Here and throughout the paper, we use a bar to denote the average over the periodic variables, i.e.,
$$
\bar f_0(k) := \frac{1}{(2\pi)^d}\int_{\T^d} f_0(q,k) \di q.
$$

The relevance of phase space mixing in clusters of stars had been pointed out historically by Lynden--Bell \cite{dL62,dL67} and was highlighted more recently in Mouhot and Villani's proof of Landau Damping in the Vlasov-Poisson system on the torus \cite{cMcV11}. 
This sparked interest in proving mixing in linear models that describe astrophysical systems \cite{pRoS18} or are reasonable \emph{toy models} for these systems \cite{pRoS18b}.
The recent paper \cite{sCjL21} studies the one-dimensional Liouville equation with a slightly anharmonic potential $ V(x)= \frac{x^2}{2} + \epsilon \frac{x^4}{2}$ and proves the time convergence of the (one-dimensional) Coulomb potential generated by this distribution. The authors of \cite{sCjL21} use the so-called vector-field method and obtain a rate of convergence. The vector-field method has been introduced by Klainerman \cite{klainerman1985uniform} in the context of wave equations and has been applied to transport equations, or the Vlasov--Poisson system, for instance in \cite{smulevici2016small, fajman2017vector, wong2018commuting}.   

\medskip
In this paper, we apply the vector-field method to general integrable systems and any choice of test function $\phi$. We obtain power-like convergence to the the equilibrium value $\int \phi \bar f_0$. 
Since we are studying essentially a transport equation in $\T^d$ rather than $\R^d$, the rate of decay does not improve with dimension. 
We will assume throughout that $\omega: K \mapsto \R^d$ is of class $C^2$ and use the notation $D \omega$ for its Jacobian matrix, i.e.,
$$
(D \omega)_{j,l}(k) = \partial_{j} \omega_l (k).  
$$
With these preliminaries in place, we can state our main theorem. 
\begin{theorem} \label{thm:main}
Let $f(t,q,k)$ be the solution to \eqref{eq:transport} with initial datum $f_0 \in C^1(\T^d\times K)$. Assume that $\phi \in C^1_c(\T^d \times K)$ is bounded, and that 
\begin{equation} \label{eq:cond_thm}
    \omega \in C^2(K), \quad \text{ and } \det D \omega (k) \neq 0,  \text{ for all } k \in K
\end{equation}
then there exists $C$ depending on $\omega$, $f_0$ and $\phi$ such that
$$
\abs{\int_{K}\int_{\T^d}
(f(t,q,k)- \bar f_0 (k) ) \phi(q,k) \di q \di k} \le \frac{C}{1+|t|}.
$$
\end{theorem}

\begin{remark}
The constant $C$ depends on the initial data, on the test function $\phi$, and on the inverse of $D\omega$. In Propositions~\ref{prop:1-d-bound} and~\ref{prop:d>1} below, we give more precise statements that allow to relax the hypotheses and estimate the constant for concrete cases.
\end{remark}
\begin{remark}
For the case of a particle in a one-dimensional potential well, $2\pi/\omega(k)$ is the inverse of the period $T$ of the trayectories, which in turn is the derivative of the area function $\Pi$, see e.g., \cite[Section 50]{Arnold-Book}. 
In terms of the potential $V$, we have
$$ 
T(h) =  \Pi' (h) , \qquad \Pi(h) := \oint_{\frac{p^2}{2} + V(x) = h} p \di x .
$$
Apart from regularity issues, condition~\eqref{eq:cond_thm} is simply $T'(h)= \Pi''(h) \neq 0$. For one-dimensional systems (2-dimensional phase space), there is an extensive literature on the monotonicity properties of the period function $T(h)$. For many potentials, $T'(h)$ has a definite sign, see e.g. \cite{cC87,sCdW86, freire2004first,fR93}. In \cite[Appendices B and C]{pRoS18b}, these conditions are specified to several potentials relevant in astrophysics.
\end{remark}

Even if the condition $\det D\omega \neq 0$ fails at some points, mixing may still hold. For simplicity, we state this result in the one-dimensional case and for a linearly vanishing $\omega'$. 

\begin{theorem}\label{thm:1-d-degenerate}
Fix $f_0$ and $\phi$ of class $C^1$, with compact support, and let $f$ denote the corresponding solution to Liouville's equation. 
Assume that $\omega \in C^2(K)$, and $\omega'(k) \neq 0$ except for $k$ in the finite set $\{ k_1, \cdots ,k_N \}$, and that $\omega''(k_i) \neq 0$. 
Then, there is $C>0$ such that
\begin{align*}
    &\abs{ \int_{K} \int_{\T} (f(t,q,k)- \bar f (k)) \phi(q,k)} \di q \di K  \le
   \frac{C}{1+|t|^{1/3}} \qedhere .
\end{align*}
\end{theorem}

In a different direction, if condition \eqref{eq:cond_thm} holds and the functions involved have a better regularity, the strategy in the proof of Theorem~\ref{thm:main} can be iterated to obtain a better rate of decay. Again, we state the one-dimensional result for simplicity.
\begin{theorem} \label{thm:1-d-improved}
For $d=1$ and under the hypotheses of Theorem~\ref{thm:main}, assume that additionally, $\omega'(k)^{-1} \in C^l(K)$, $f_0, \phi \in C^l(\T\times K)$ for some $l\ge 2$. 
Then there exists $C>0$ depending on $\omega$, $f_0$ and $\phi$ such that
$$
\abs{\int_{K}\int_{\T}
(f(t,q,k)- \bar f_0 (k) ) \phi(q,k) \di q \di k} \le \frac{C}{1+|t|^l}.
$$
\end{theorem}

A striking consequence is that mixing is actually super-polynomial when $\omega$, $f_0$ and $\phi$ are of class $C^\infty$.

\medskip
Finally, we study the Coulomb potential generated by a particle density $F$. We will use the notation $F$ for the density in the physical phase space $\R^d \times \R^ d$ and $f = F \circ N$ for the density in action-angle coordinates. The motivation to consider the Coulomb potential in particular, is to take into account the gravitational self-interaction (the Vlasov--Poisson system). As in \cite{sCjL21}, the results that we prove remain insufficient to treat the nonlinear equation. This is natural, since we don't expect in general that $\bar f_0$ is a stationary state for the Vlasov--Poisson system.

The Coulomb potential can be written as the integral of $F$ against a \emph{test function} with a singularity, which can be compensated by requiring some extra regularity of $F$. 
For a given $F$ defined in Euclidean space, we define the Coulomb potential generated by its particle density as the unique solution to 
\begin{equation}
    \label{eq:def-coulomb}
    -\Delta V_F (x) = \int_{\R^d} F(x,p) \di p,  \text{ and } \begin{cases} V_F(0) = 0 &\text{ if } d=1 \\
\lim_{|x|\to \infty} V_F(x) = 0 &\text{ if } d\ge 2 .
\end{cases}
\end{equation}
We will assume that the system with Hamiltonian $H(x,p)= |p|^2/2 + V(x)$ is integrable and denote  
$N: \T^d \times K \mapsto G \subset \R^d \times \R^d$ for the transformation from action-angle variables to the position and momentum, where $G$ is the open set of values of position and momenta for which this transformation is well-defined and invertible. Then we have the following corollary.

\begin{corollary}\label{cor:coulomb}
 Assume that $N$ is a $C^1$-diffeomorphism, and that the frequencies $\omega(k)$ satisfy \eqref{eq:cond_thm}.
 Let $F_0 \in C^1_c(G) \cap L^1(G)$ .
Denote by $F$ the solution to Liouville's equation \eqref{eq:liouville}, then
$$
\norm{V_F -V_{\widetilde{F}_0}}_{L^\infty} \le \frac{C}{1+ |t|},
$$
where 
$$
\widetilde F_0  := \overline{(F_0\circ N)}\circ N^{-1} .
$$
\end{corollary}

The remainder of this paper is organized as follows. In section~\ref{sec:1-d}, we prove the one-dimensional case of Theorem~\ref{thm:main}, and Theorems~\ref{thm:1-d-degenerate}~and~\ref{thm:1-d-improved}. In section~\ref{sec:higher-d}, we prove the general case of Theorem~\ref{thm:main} and its Corollary~\ref{cor:coulomb}.

\section{The one-dimensional case.} \label{sec:1-d}

The main tool in the proof of this theorem is the \emph{vector field} 
\begin{equation} \label{eq:def_W}
    W := \omega'(k) t \partial_q + \partial_k.
\end{equation}

A straightforward calculation shows that $W$ commutes with the Liouville operator
$$
L:= \partial_t + \omega(k) \partial_q.
$$
Therefore, if $f$ solves the transport equation \eqref{eq:transport}, the same goes for $W^n f$ (and $|W^n f|$), for any $n = 0,1,2, \dots$ and thus, 
\begin{equation} \label{eq:conservation}
    \iint \abs{W^n f}(t,q,k) g(k) \di k \di q = \iint \abs{W^n f_0}(q,k) g(k) \di k \di q,
\end{equation}
for sufficiently regular functions $f$ and $g$.
As usual, we will use this property to obtain time-indepent bounds.

\begin{proposition} \label{prop:1-d-bound}
Let $f$ denote the solution to \eqref{eq:transport} with initial data $f_0 \in L^1$ and fix $\phi \in L^\infty$. Assume that either $f$ or $\phi$ have compact support in $\T^d \times K$. 
Then, provided all terms on the right-hand-side are finite,
\begin{align*}
  &  \Bigl|\int_{
  K}\int_{\T} 
(f(t,q,k)- \bar f_0 (k) ) \phi(q,k) \di q \di k \Bigr| \\
&\qquad \le \frac{2\pi }{t}\left(
\norm{ \frac{\bar \phi}{\omega'} \partial_k f_0 }_{L^1}  + \norm{\bar f_0\partial_k \frac{ \phi}{\omega'}  }_{L^1} 
\right) \numberthis \label{eq:1-d-bound}
\end{align*}

\end{proposition}
\begin{remark}
The hypotheses on $f$, $\phi$ and $\omega$ of Theorem~\ref{thm:main} imply directly that the terms in the upper bound are indeed finite. Since it is sufficient to prove the decay for large values of $t$, this proposition implies the one-dimensional case of Theorem~\ref{thm:main}. 
\end{remark}
\begin{remark}
The hypothesis on compact support is only needed to ensure the absence of boundary terms when integrating by parts. It can be weakened by adding the value(s) of $\frac{\bar f_0 \bar \phi}{\omega'}$ at $\partial K$ to the right-hand-side, provided these values are well-defined.
\end{remark}

\begin{proof}
First, note that as a limiting case of \eqref{eq:conservation} or by using the exact time evolution and a change of variables
\begin{align*}
    \bar f_0(k) 
    &:= \frac{1}{2\pi} \int_{\T} f_0(q,k) \di q  \\
    &= \frac{1}{2\pi} \int_{\T} f_0(q- \omega(k) t,k) \di q =: \frac{1}{2\pi} \int_{\T} f(t,q,k) \di q.
\end{align*}
We insert this in the expression that we need to estimate and use the fundamental theorem of calculus to write
\begin{align*}
  \Bigl|\int_{K}\int_{\T} 
(f(t,q,k)- & \bar f_0 (k) ) \phi(q,k) \di q \di k \Bigr| \\
&=
\left| \frac{1}{2 \pi} \int_{K}\int_{\T}\int_{\T}
(f(t,q,k)-  f (t,q',k) ) \phi(q,k) \di q'\di q \di k  \right| \\
 & = \left|\frac{1}{2 \pi} \int_{K}\int_{\T}\int_{\T} \int_{q'}^{q} 
\partial_q f(t,\tilde q,k) \phi(q,k) \di \tilde q \di q'\di q \di k \right| \\
& \le 
\frac{1}{2 \pi} \int_\T \int_\T \int_\T\left|
\int_{K} \partial_q  f(t,\tilde q,k) \phi(q, k ) \di k
\right| \di \tilde q \di q' \di q. \\
&=   \int_\T \int_\T\left|
\int_{K} \left[\partial_q  f(t,\tilde q,k)\right] \phi(q, k ) \di k
\right| \di \tilde q  \di q. 
\end{align*}
To obtain this inequality, we first used Fubini's theorem to perform the $k$-integral before the others, and then extended the range of $\tilde q$ (which requires inserting the absolute value). The last line is just the observation that the $q'$-dependence has disappeared from the integrand.

We now use $W$ defined in \eqref{eq:def_W} to write $\partial_q = (\omega'(k) t)^{-1}(W - \partial_k)$. The first term will have the required form to apply \eqref{eq:conservation}, and we can integrate by parts (the boundary terms dissappear due to the assumptions on $f$ and $\phi$) to bring the second term in this form as well. This gives
\begin{align*}
&\int_{K} \left[\partial_q  f(t,\tilde q,k) \right]\phi(q, k ) \di k  \\
& \qquad =
t^{-1} \int_{K}\left[W  f\right](t,\tilde q,k)  \frac{ \phi(q, k )}{\omega'(k)} \di k
+ t^{-1} \int_{K} f(t,\tilde q,k) \partial_k \left[\frac{ \phi(q, k )}{\omega'(k)} \right] \di k.
\end{align*}

Inserting this in the bound, we found 
\begin{align*}
    \Bigl|\int_{K}\int_{\T} 
(f(t,q,k)- & \bar f_0 (k) ) \phi(q,k) \di q \di k \Bigr| \\
& \le   t^{-1} \int_{K} \int_\T \frac{\abs{Wf}(t,\tilde q , k)}{|\omega'(k)|} \int_{\T} \abs{\phi(q,k)} \di q \di \tilde q \di k \\
& \qquad +  t^{-1} \int_{K} \int_\T f(t, \tilde q, k ) \int _\T \abs{\partial_k \frac{ \phi(q, k )}{\omega'(k)}} \di q \di \tilde q \di k .
\end{align*}

In both terms we can apply \eqref{eq:conservation}, and thus, using furthermore that $W\bigr|_{t=0}= \partial_k$, we have obtained
\begin{align*}
    \Bigl|\int_{K}\int_{\T} 
(f(t,q,k)- & \bar f_0 (k) ) \phi(q,k) \di q \di k \Bigr| \\
& \le   t^{-1} \int_{K} \int_\T \frac{\abs{\partial_k f_0}(\tilde q , k)}{|\omega'(k)|} \int_{\T} \abs{\phi(q,k)} \di q \di \tilde q \di k \\
& \qquad +  t^{-1} \int_{K} \int_\T f_0( \tilde q, k ) \int _\T \abs{\partial_k \frac{ \phi(q, k )}{  \omega'(k) }} \di q \di \tilde q \di k .
\end{align*}

This can be rewritten in terms of the averages over $\T$ to give \eqref{eq:1-d-bound}.
\end{proof}

\subsection{Localization argument}
In this section we prove Theorem~\ref{thm:1-d-degenerate}. We use the explicit rate of decay and the expression for the upper bound obtained in Proposition~\ref{prop:1-d-bound} allows for extensions when $\omega'(k)$ vanishes at some \emph{energies} in the support of $\phi$. 
We use a simple localization argument to treat the case where $\omega'(k)$ vanishes linearly.
\begin{proof}[Proof of Theorem~\ref{thm:1-d-degenerate}]
Let $0<\epsilon < 1$ to be fixed later. We fix a smooth cutoff function $\chi$ with support in $(-1,1)$, values in $ [0,1]$, and such that $\chi \equiv 1$ in $[-1/2, 1/2]$.
We define $\chi_{i,\epsilon}:= \chi(\frac{k -k_i}{\epsilon})$ and
$\eta_\epsilon := \prod_{i=1}^N (1 - \chi_{i,\epsilon})$. 
Then, we write
$\phi(q,k) = \eta_\epsilon(k) \phi(q,k) + (1-\eta_{\epsilon}(k)) \phi(q,k) $.
Note that $\eta_\epsilon(k) \phi(q,k)$ satisfies the hypotheses of Proposition~\ref{prop:1-d-bound}.
Thus, \eqref{eq:1-d-bound} gives
\begin{align*}
    &\abs{ \int_{K } \int_{\T} (f(t,q,k)- \bar f_0 (k))\eta_\epsilon(k) \phi(q,k) \di q \di k} \\
    &\qquad \le
    \frac{2 \pi  }{t}\norm{\frac{\eta_\epsilon}{ \omega'}}_{L^\infty} \left( \norm{\bar\phi \partial_k  f_0}_{L^1} +  \norm{\bar f_0\partial_k \phi  }_{L^1}\right) \\
  & \qquad \qquad +\frac{  2 \pi  }{t}\norm{\partial_k \left(\frac{\eta_\epsilon}{ \omega'}\right)}_{L^\infty} \norm{\phi  \bar f_0}_{L^1}.  
\end{align*}

Now, we need to extract the $\epsilon$-dependence from the $L^\infty$-norms. Since $\omega''(k_i) \neq 0$, for some $C >0$ and all $\epsilon \le  1$, we have
$$
\inf_{\supp (\eta_\epsilon) \cap \supp (f)} |\omega'(k)| \ge\frac{\epsilon}{C}.
$$ 
This gives the bounds
$$
\norm{\frac{\eta_\epsilon}{ \omega'}}_{L^\infty}  \le \frac{C}{\epsilon}, 
\qquad \norm{\partial_k \left(\frac{\eta_\epsilon}{ \omega'}\right)}_{L^\infty}\le \frac{C}{\epsilon^2}.
$$
We have obtained
\begin{equation} \label{eq:bound_regular_part}
\abs{ \int_{K } \int_{\T} (f(t,q,k)- \bar f (k))\eta_\epsilon(k) \phi(q,k)} 
    \le \frac{C}{t \epsilon^2} .
\end{equation}
On the other hand,
\begin{align*}
    &\abs{ \int_{K } \int_{\T} (f(t,q,k)- \bar f_0 (k))(1-\eta_\epsilon(k)) \phi(q,k)} \\
    &\qquad\le  2 \int_{\R^+ }  (1-\eta_\epsilon(k))  \int_{\T} f(t,q,k) \di q \di k \\
    & \qquad \le C \epsilon  \norm{\bar f_0}_{L^\infty} .
\end{align*}
We sum with \eqref{eq:bound_regular_part}, evaluate at some $T >1$  and pick $\epsilon = T^{-1/3}$, to obtain
\begin{align*}
    &\abs{ \int_{K } \int_{\T} (f(T,q,k)- \bar f (k)) \phi(q,k)}   \le
   C T^{-1/3},
\end{align*}
since for small $T$, both terms are bounded, this implies the result.
\end{proof}

\subsection{Improved decay}
If the initial condition is more regular, we can improve the estimate on the decay. 
To this end, we use the following $L^1$-version of Poincaré's inequality.
\begin{lemma}[Poincaré's inequality] \label{lem:poincare}
Assume that $g: \T \mapsto \R$ is a periodic function of class $C^l$ and $g(x)=0$ for some $x \in [0,2\pi)$. Then, for all $l \in \N$,
$$
\int_\T |g(q)| \di q  \le (\pi)^l \int_\T |g^{(l)}(q)| \di q.
$$
\end{lemma}
\begin{proof}
Without loss of generality, we may assume that $x =0$.
Then,
\begin{align*}
   \int_0^\pi  |g(q)| \di q 
   &= \int_0^\pi \abs{\int_0^s g'(r) \di r} \di s \\
   & \le \int_0^\pi \abs{g'(r)}\int_r^\pi \di s \di r \le \pi  \int_0^\pi \abs{g'(r)} \di r.
\end{align*}
Treating the contribution to the $L^1$-norm of the interval $[\pi, 2 \pi]$ analogously, we find that
$$
\int_\T |g(q)| \di q  \le \pi \int_\T |g'(q)| \di q.
$$
For the case $l\ge 2$, we proceed by induction. By periodicity $\int_{\T}g^{(l-1)}=0$, so $g^{(l-1)}(x) = 0$ for some $x \in \T$, and we can iterate the argument.
\end{proof}

As a consequence, we can obtain a faster rate of decay for more regular initial data and observables. For the sake of readability, we assume that the support of is $\phi$ compact (bounded away from the boundary of $K$), though it is possible to relax this to suitable decay of the functions and their derivatives.

\begin{proof}[Proof of Theorem~\ref{thm:1-d-improved}]
As in the proof of Proposition~\ref{prop:1-d-bound}, we bound
\begin{align*}
  \Bigl|\int_{K}\int_{\T} 
(f(t,q,k)- & \bar f_0 (k) ) \phi(q,k) \di q \di k \Bigr| \\
&\le  \int_\T \int_\T\left|\partial_{\tilde q} 
\int_{K}  f(t,\tilde q,k) \phi(q, k ) \di k
\right| \di \tilde q  \di q. 
\end{align*}
We then use Lemma~\ref{lem:poincare}
to insert $l-1$ additional derivatives:
\begin{align*}
  \Bigl|\int_{K}\int_{\T} 
(f(t,q,k)- & \bar f_0 (k) ) \phi(q,k) \di q \di k \Bigr| \\
&\le  \pi^{l-1}\int_\T \int_\T\left|\partial_{\tilde q}^{l} 
\int_{K}  f(t,\tilde q,k) \phi(q, k ) \di k
\right| \di \tilde q  \di q \\
&\le t^{-l} \pi^{l-1}\int_\T \int_\T\left|
\int_{K} \left(\frac{W - \partial_k}{\omega'(k)}\right)^{l}  f(t,\tilde q,k) \phi(q, k ) \di k
\right| \di \tilde q  \di q.
\end{align*}
In the previous expression, we keep in mind that the operator $W$ only affects the variables denoted by  $k$ and $\tilde q$, not $q$.
Expanding the product makes appear $2^l$ terms. In order to integrate by parts, we iterate the identities
\[
[\partial_K, W] = \omega''(k) t \partial_q =\frac{\omega''(k)}{\omega'(k)} \left( W- \partial_k \right)
\]
and for any sufficiently regular function $g(k)$,
\[
[W, g(k)]= [\partial_K, g(k)] = g'(k).
\]
This allows to obtain an identity of the form
\[
\left(\frac{W-\partial_k}{\omega'(k)}\right)^{l} = \sum_{j=0}^{l}\sum_{m=0}^{l} g_{j,m}(k) \partial_k^{j} W^{m},
\]
where each of the functions $g_{n,m}(k)$ is a complicated expression containing powers of $(\omega')^{-1}$ and its derivatives up to order $l-(m+j)$.
In each term, we integrate by parts in $K$ to obtain
\begin{align*}
    \left| 
\int_{K} \left(\frac{W-\partial_k}{\omega'(k)}\right)^{l} f(t,\tilde q,k) \phi(q, k ) \di k
\right| 
&\le
\sum_{j=0}^{l}\sum_{m=0}^{l}  \left|\int_K W^m f(t,\tilde q, k) \partial_k^j (g_{j,m} (k) \phi(k))\di k\right|. \\
&\le \sum_{j=0}^{l}\sum_{m=0}^{l} \norm{\partial_k^j (g_{j,m} \phi)}_{L^\infty} \int_K\left| W^m f(t,\tilde q, k)\right| \di k.
\end{align*}
Finally, we apply $\eqref{eq:conservation}$ to bound
\begin{align*}
  \Bigl|\int_{\R_+}\int_{\T} 
(f(t,q,k)- & \bar f_0 (k) ) \phi(q,k) \di q \di k \Bigr| \\
&\le t^{-l} \pi^{l} \sum_{j=0}^{l}\sum_{m=0}^{l} \norm{\partial_k^j (g_{j,m} \phi)}_{L^\infty} \int_\T\int_K\left| W^m f(t,\tilde q, k)\right| \di \tilde q  \\
&= t^{-l} \pi^{l} \sum_{j=0}^{l}\sum_{m=0}^{l} \norm{\partial_k^j (g_{j,m} \phi)}_{L^\infty} \norm{\partial_k^m f_0}_{L^1} \qedhere
\end{align*}
\end{proof}

\section{Proofs for $d\ge 2$} \label{sec:higher-d}

In this section we prove mixing for the $d$-dimensional case.
In this case, we define the $d$ vector fields
\begin{equation} \label{eq:def_W_d_dimensions}
    W = t D\omega(k) \nabla_q + \nabla_k \qquad \text{ or } \quad W_j = t \sum_{i=1}^d ( \partial_{k_j}\omega_i(k)) \partial_{q_i} + \partial_{k_j}. 
    \end{equation}

As before, $W$ commutes with the Liouville operator.
If $D\omega$ is an invertible matrix, most of the proof goes through as before. For the sake of completeness, we state Theorem~\ref{thm:main} with an explicit bound on the right-hand-side. 
To this end, we define the matrix norm
$$
|M|_{\infty} = \max_{i,j} |M_{i,j}|.
$$

\begin{proposition} \label{prop:d>1}
Let $f(t,q,k)$ be the solution to \eqref{eq:transport} with initial datum $f_0 \in C^1(\T^d\times K)$. Assume that $\phi \in C^1_c(\T^d \times K)$, and that 
\begin{equation*}
    \det D\omega (k) \neq 0, \quad \text{ for all } k \in K.
\end{equation*}
Then, $M : = (D\omega(k))^{-1}$ is well-defined and
\begin{align*}
  &  \abs{\int_{K}\int_{\T}
(f(t,q,k)- \bar f_0 (k) ) \phi(q,k) \di q \di k}  \\
&\quad \le 
 \frac{2\pi d}{t} \left( \norm{|M|_{\infty} \phi}_{L^\infty} \sum_{i=1}^d \norm{\partial_{k_i} \bar f_0}_{L^1} +\norm{\nabla_k \cdot M \phi \hat{e_j}}_{L^\infty} \norm{\bar f_0}_{L^1}  \right).
\end{align*}
\end{proposition}

\begin{proof}
As before, we express the left-hand-side as
\begin{align*}
     & \Bigl|\int_{K}\int_{\T^d} 
(f(t,q,k)-  \bar f_0 (k) ) \phi(q,k) \di q \di k \Bigr| \\
&\qquad= \frac{1}{(2\pi)^d} \abs{ \int_{\T^d } \int_{\T^d } 
\int_{K} (f(t,q,k) - f(t,q',k)) \phi(q,k) \di k \di q \di q'}.
\end{align*}
We then write
\begin{align*}
    f(t,q,k) - f(t,q',k) = \sum_{j=1}^d \int_{q'_j}^{q_j} 
    \partial_{q_j} f(t,q_{1}, \cdots q_{j-1}, s, q'_{j+1}, \cdots, q'_d , k ) \di s
\end{align*}
and bound
\begin{align*}
     & \Bigl|\int_{K}\int_{\T^d} 
(f(t,q,k)-  \bar f_0 (k) ) \phi(q,k) \di q \di k \Bigr| \\
&\le \frac{1}{(2\pi)^d}\sum_{j=1}^d  \int_{\T^d } \int_{\T^d } \int_0^{2\pi} \abs{
\int_{K} \partial_{q_j} f(t,q_{1}, \cdots q_{j-1}, s, q'_{j+1}, \cdots, q'_d , k ) \phi(q,k) \di k} \di s \di q \di q'.
\end{align*}
We now write
$\partial_{q_j} = t^{-1} [M(W - \nabla_k)]_j $.
Thus, by using the divergence theorem for the second term (and using the compact support of $\phi$ to conclude the absence of boundary terms), we obtain
\begin{align*}
  &  \abs{\int_{K} \partial_{q_j} f(t,q_{1}, \cdots q_{j-1}, s, q'_{j+1}, \cdots, q'_d , k ) \phi(q,k) \di k} \\
  &\quad  \le \frac{1}{t} \int_{K}\abs{ [M(k) W]_j f(t,q_{1}, \cdots q_{j-1}, s, q'_{j+1}, \cdots, q'_d , k )\phi(q,k) }\di k  \\
  &\qquad +  \frac{1}{t} \int_{K}\abs{  f(t,q_{1}, \cdots q_{j-1}, s, q'_{j+1}, \cdots, q'_d , k )(\nabla_k \cdot M(k)^\top \phi(q,k) \hat{e_j})}\di k \\
  &\quad \le \frac{\norm{|M|_{\infty} \phi}_{L^\infty}}{t} \sum_{i=1}^d\int_{\R_+^d}\abs{ W_i f(t,q_{1}, \cdots q_{j-1}, s, q'_{j+1}, \cdots, q'_d , k ) }\di k \\
   &\qquad +  \frac{\norm{\nabla_k \cdot M^\top \phi \hat{e_j}}_{L^\infty}}{t} \int_{K}\abs{  f(t,q_{1}, \cdots q_{j-1}, s, q'_{j+1}, \cdots, q'_d , k )}\di k.
\end{align*}
Inserting this in the previous bound and using \eqref{eq:conservation}, we obtain
\begin{align*}
     & \Bigl|\int_{K}\int_{\T^d} 
(f(t,q,k)-  \bar f_0 (k) ) \phi(q,k) \di q \di k \Bigr| \\
&\le \frac{2\pi d}{t} \left( \norm{|M|_{\infty} \phi}_{L^\infty} \sum_{i=1}^d \norm{\partial_{k_i} \bar f_0}_{L^1} +\norm{\nabla_k \cdot M^\top \phi \hat{e_j}}_{L^\infty} \norm{\bar f_0}_{L^1}  \right) \qedhere
\end{align*}
\end{proof}

Finally, we prove Corollary~\ref{cor:coulomb}.
\begin{proof}[Proof of Corollary~\ref{cor:coulomb}]
For fixed $x_0 \in \R^d$, we write $\phi_{x_0}$ for the fundamental solution to Poisson's equation in dimension $d$. 
In particular, 
$$
\phi_{x_0}(x) \begin{cases}
(x-x_0)\id_{[-\infty,x_0]} + \max(0, x_0) &\text{ if } d=1 \\
-(2\pi)^{-1}\ln (|x-x_0|)  &\text{ if } d= 2 \\
\kappa_d |x-x_0|^{-d+2} &\text{ if } d= 3 ,
\end{cases}
$$
for a suitable constant $\kappa_d$.
So we can write
\begin{align*}
    V_F(x_0) = \int_{\R^d} \int_{\R^d} \phi_{x_0} (x) F(t,x,v) \di x \di v 
    &=  \int_{\T^d} \int_{K} \varphi_{x_0} (q,k)  f(t,q,k) \di k \di q,
    \end{align*}
   where we used the notation $ f = F\circ N$ and $ \varphi_{x_0} = \phi_{x_0} \circ N$. 
Now, the integral is in a suitable form to apply the arguments in the proofs of Theorems~\ref{prop:1-d-bound} and \ref{prop:d>1}, provided that $\varphi_{x_0}$ is sufficiently regular. The Coulomb kernel $\phi_{x_0}$ belongs to the Sobolev space $ W^{1, 1}_{\rm loc}(\R^d \times \R^d)$, since the integral of its derivative in a ball is finite. Outside a sufficiently large ball, the function and its derivatives are bounded. Since we assume that $N$ is of class $C^1$, $\varphi_{x_0}$ inherits these properties. Thus, $\nabla_k  \varphi_{x_0} \in L^1 + L^\infty$. For the $L^\infty$-part, we can apply Proposition~\ref{prop:d>1} directly, and for the $L^1$ part we switch the roles of $f$ and $\phi$ in the proof of Proposition~\ref{prop:d>1}.
\end{proof}

\end{document}